\documentclass[graybox]{svmult}

\usepackage{mathptmx}       
\usepackage{helvet}         
\usepackage{courier}        
\usepackage{type1cm}        
\usepackage{makeidx}         
\usepackage{graphicx}        
\usepackage{multicol}        
\usepackage[bottom]{footmisc}

\begin{document}

\title*{On a pair of difference equations for the $_4F_3$ type orthogonal polynomials and related exactly-solvable quantum systems}
\titlerunning{A pair of difference equations for the $_4F_3$ type orthogonal polynomials}
\author{E.I. Jafarov, N.I. Stoilova and J. Van der Jeugt}
\institute{E.I. Jafarov \at Institute of Physics, Azerbaijan National Academy of Sciences, Javid av. 131, AZ-1143 Baku, Azerbaijan, \email{ejafarov@physics.ab.az}
\and N.I. Stoilova \at Institute for Nuclear Research and Nuclear Energy, Bulgarian Academy of Sciences, Boul. Tsarigradsko Chaussee 72, 1784 Sofia, Bulgaria \email{stoilova@inrne.bas.bg}
\and J. Van der Jeugt \at Department of Applied Mathematics, Computer Science and Statistics, University of Ghent, Krijgslaan 281-S9, B-9000 Gent, Belgium \email{Joris.VanderJeugt@UGent.be}}
\maketitle

\abstract{
We introduce a pair of novel difference equations, whose solutions are expressed in terms of Racah or Wilson polynomials depending on the nature of the finite-difference step. 
A number of special cases and limit relations are also examined, which allow to introduce similar difference equations for the orthogonal polynomials of the $ _3 F_2$ and $ _2 F_1$ types. 
It is shown that the introduced equations allow to construct new models of exactly-solvable quantum dynamical systems, such as spin chains with a nearest-neighbour interaction and fermionic quantum oscillator models. 
}

\section{Introduction}
\label{int}

The importance of orthogonal polynomials in the study of quantum dynamical systems is undisputable. 
Without the knowledge of basic properties of orthogonal polynomials, it is impossible to comprehend the existence of explicit solutions of quantum systems such as the quantum harmonic oscillator, the Coulomb problem or Heisenberg spin chains. 
A long time ago, different types of orthogonal polynomials were studied separately.
Then the idea grew that some of them are special case of others, and that they can be generalized.
Thus the discovered polynomials could be unified in a table, each having its own level and cell in that table. 
This table is called the Askey scheme of hypergeometric orthogonal polynomials. 
The importance of this table is that it gathers all polynomials, some of them satisfying an orthogonality relation in the continuous space and others in a discrete space, some with a finite support and others with an infinite support~\cite{koekoek}.

Hermite polynomials are the most attractive ones from the Askey scheme, because they have no free parameters and occupy the lowest level of the table, that is the level where there is no sign of the discreteness of the space. 
They are well known as the explicit solution of the 1D non-relativistic quantum harmonic oscillator in a canonical approach~\cite{landau}. 
The dynamical symmetry of this quantum system is also well known and it is the Heisenberg-Weyl algebra. 
This algebra can be easily constructed by using the three-term recurrence relations of Hermite polynomials. 
If, as a next step, one drops the canonical commutation relation between position and momentum operator $\left[ {\hat p,\hat x} \right] =  - i$~\cite{wigner}, then one observes the very interesting behaviour of the solution of the 1D non-relativistic quantum harmonic oscillator. 
Now the solution is expressed in terms of the generalized Laguerre polynomials, and the dynamical symmetry of the system is the Lie superalgebra $osp\left( {1|2} \right)$.
It is constructed by using two kind of three-term recurrence relations of the generalized Laguerre polynomials, which are intertwined.
The existence of more than one recurrence relations for these polynomials has the following explanation. 
Laguerre polynomials occupy the next level in the Askey scheme: they generalize Hermite polynomials and have one free parameter.
This parameter allows to separate the recurrence relations for even and odd polynomials, and thus obtain the new form of the recurrence relations for generalized Laguerre polynomials, which leads to the quite interesting so called non-canonical solution of the 1D non-relativistic quantum harmonic oscillator~\cite{ohnuki}. 
It is known that such a method can also be applied to polynomials from the next levels of the Askey scheme, and similar recurrence relations exist for continuous dual Hahn polynomials~\cite{groenevelt}, generalizing both Meixner-Pollaczek and Laguerre polynomials.
Their application allows to construct a new model of the quantum harmonic oscillator, whose algebra is the Lie algebra $su(1,1)$ deformed by a reflection operator~\cite{jafarov1}. 
A similar approach in finite-discrete configuration space leads to the new difference equations (or recurrence relations) for the Hahn or dual Hahn polynomials and they generalize the difference equation for Krawtchouk polynomials (due to duality of Krawtchouk polynomials, the difference equation can be transformed to the three-term recurrence relation). 
Application of such recurrence relations leads to two very interesting quantum mechanical solutions, one of which is a finite-discrete quantum oscillator model based on the Lie algebra $u(2)$ extended by a parity operator~\cite{jafarov2} and other one is the case of perfect state transfer over the spin chain of fermions with a nearest-neighbour interaction under absence of the external magnetic field~\cite{stoilova}.

In current work, we continue this procedure and report on the pairs of three-term difference equations and recurrence relations for the Racah and Wilson polynomials, which occupy the top level of the Askey scheme and generalize all discrete and continuous orthogonal polynomials from this table. 
We also discuss some special cases, when new three-term difference equations exist also for Hahn polynomials and they lead to a pair of difference equations for the continuous Hahn polynomials.

\section{Racah polynomials and new three-term recurrence relations}
\label{racah}

The Racah polynomial $R_n \left( {\lambda \left( x \right);\alpha ,\beta ,\gamma ,\delta } \right)$ of degree $n$ ($n = 0,1, \ldots ,m$) in the variable $x$ is defined by:
\begin{equation}
\label{racah-def}
R_n \left( {\lambda \left( x \right);\alpha ,\beta ,\gamma ,\delta } \right) = {\kern 1pt} _4 F_3 \left( {\begin{array}{*{20}c}
   {\begin{array}{*{20}c}
   { - n,n + \alpha  + \beta  + 1, - x,x + \gamma  + \delta  + 1}  \\
   {\alpha  + 1,\beta  + \delta  + 1,\gamma  + 1}  \\
\end{array};} & 1  \\
\end{array}} \right),
\end{equation}
where $\lambda \left( x \right) = x\left( {x + \gamma  + \delta  + 1} \right)$ and $\alpha  + 1=-m$ or $\beta  + \delta  + 1=-m$ or $\gamma  + 1=-m$, with $m$ being a nonnegative integer.

They satisfy a finite-discrete orthogonality relation of the following form:
\begin{equation}
\label{racah-ort}
\sum\limits_{x = 0}^m {w\left( x \right)R_l \left( {\lambda \left( x \right);\alpha ,\beta ,\gamma ,\delta } \right)R_n \left( {\lambda \left( x \right);\alpha ,\beta ,\gamma ,\delta } \right)}  = h_n \delta _{ln} ,
\end{equation}
where
\begin{equation}
\label{racah-weight}
w\left( x \right) = \frac{{\left( {\alpha  + 1} \right)_x \left( {\beta  + \delta  + 1} \right)_x \left( {\gamma  + 1} \right)_x \left( {\gamma  + \delta  + 1} \right)_x \left( {\left( {\gamma  + \delta  + 3} \right)/2} \right)_x }}{{\left( { - \alpha  + \gamma  + \delta  + 1} \right)_x \left( { - \beta  + \gamma  + 1} \right)_x \left( {\left( {\gamma  + \delta  + 1} \right)/2} \right)_x \left( {\delta  + 1} \right)_x x!}},
\end{equation}
\begin{equation}
\label{racah-norm}
h_n  = M \cdot \frac{{\left( {n + \alpha  + \beta  + 1} \right)_n \left( {\alpha  + \beta  - \gamma  + 1} \right)_n \left( {\beta  + 1} \right)_n n!}}{{\left( {\alpha  + \beta  + 2} \right)_{2n} \left( {\alpha  + 1} \right)_n \left( {\beta  + \delta  + 1} \right)_n \left( {\gamma  + 1} \right)_n }},
\end{equation}
and with multiplier $M$ being defined as
\[
M = \left\{ \begin{array}{l}
 \frac{{\left( { - \beta } \right)_m \left( {\gamma  + \delta  + 2} \right)_m }}{{\left( { - \beta  + \gamma  + 1} \right)_m \left( {\delta  + 1} \right)_m }}\quad \quad \;\;\textrm{if}\quad \alpha  + 1 =  - m \\ 
 \frac{{\left( { - \alpha  + \delta } \right)_m \left( {\gamma  + \delta  + 2} \right)_m }}{{\left( { - \alpha  + \gamma  + \delta  + 1} \right)_m \left( {\delta  + 1} \right)_m }}\quad \;\;\textrm{if}\quad \beta  + \delta  + 1 =  - m \\ 
 \frac{{\left( {\alpha  + \beta  + 2} \right)_m \left( { - \delta } \right)_m }}{{\left( {\alpha  - \delta  + 1} \right)_m \left( {\beta  + 1} \right)_m }}\quad \quad \quad \textrm{if}\quad \gamma  + 1 =  - m. \\ 
 \end{array} \right.
\]

Then, one can introduce a pair of new difference equations for these polynomials in which Racah polynomials of the same degree $n$ in variables $x$ or $x+1$, and with parameters of type $(\alpha+1,\beta-1,\delta)$ and $(\alpha,\beta,\delta-1)$ are intertwined.

\begin{proposition}
The Racah polynomials satisfy the following difference equations:
\begin{eqnarray}
&& \frac{{\left( {x + \gamma  + 1} \right)\left( {x + \beta  + \delta } \right)}}{{2x + \gamma  + \delta  + 1}}R_n \left( {\lambda \left( {x + 1} \right);\alpha ,\beta ,\gamma ,\delta  - 1} \right) - \frac{{\left( {x - \beta  + \gamma  + 1} \right)\left( {x + \delta } \right)}}{{2x + \gamma  + \delta  + 1}} \nonumber \\ 
&&	\times R_n \left( {\lambda \left( x \right);\alpha ,\beta ,\gamma ,\delta  - 1} \right)= \frac{{\left( {n + \alpha  + 1} \right)\left( {n + \beta } \right)}}{{\alpha  + 1}}R_n \left( {\lambda \left( x \right);\alpha  + 1,\beta  - 1,\gamma ,\delta } \right), \nonumber \\ 
&& \label{racah-eq1} \\
&& \frac{{\left( {x + \alpha  + 2} \right)\left( {x + \gamma  + \delta + 1} \right)}}{{2x + \gamma  + \delta  + 2}}R_n \left( {\lambda \left( {x + 1} \right);\alpha  + 1,\beta  - 1,\gamma ,\delta } \right) \nonumber \\
&&\qquad - \frac{{\left( {x + 1} \right)\left( {x - \alpha  + \gamma  + \delta } \right)}}{{2x + \gamma  + \delta  + 2}} 
R_n \left( {\lambda \left( x \right);\alpha  + 1,\beta  - 1,\gamma ,\delta } \right) \nonumber \\ 
  \label{racah-eq2} 
&&  = \left( {\alpha  + 1} \right)R_n \left( {\lambda \left( {x + 1} \right);\alpha ,\beta ,\gamma ,\delta  - 1} \right). 
\end{eqnarray}
\end{proposition}

\begin{proof}
We prove both equations by performing straightforward computations using known properties of hypergeometric functions and Pochhammer symbols. 
In the case of (\ref{racah-eq1}), one can rewrite the left-hand side in the following form:
\begin{eqnarray}
&& \left( {x + \gamma  + 1} \right)R_n \left( {\lambda \left( {x + 1} \right);\alpha ,\beta ,\gamma ,\delta  - 1} \right) - \left( {x - \beta  + \gamma  + 1} \right)R_n \left( {\lambda \left( x \right);\alpha ,\beta ,\gamma ,\delta  - 1} \right) \nonumber \\ 
&&  + \frac{{\left( {x + \gamma  + 1} \right)\left( {x - \beta  + \gamma  + 1} \right)}}{{2x + \gamma  + \delta  + 1}} [ R_n \left( {\lambda \left( x \right);\alpha ,\beta ,\gamma ,\delta  - 1} \right) \nonumber \\
&& \qquad\qquad - R_n \left( {\lambda \left( {x + 1} \right);\alpha ,\beta ,\gamma ,\delta  - 1} \right) ]. \label{racah-eq11}
\end{eqnarray}
Then, a simple computations show that
\begin{eqnarray}
&& \left( {x + \gamma  + 1} \right)R_n \left( {\lambda \left( {x + 1} \right);\alpha ,\beta ,\gamma ,\delta  - 1} \right) - \left( {x - \beta  + \gamma  + 1} \right)R_n \left( {\lambda \left( x \right);\alpha ,\beta ,\gamma ,\delta  - 1} \right) \nonumber \\ 
&&  =  - \sum\limits_{k = 0}^n {\frac{{\left( { - n} \right)_k \left( {n + \alpha  + \beta  + 1} \right)_k \left( { - x} \right)_{k - 1} \left( {x + \gamma  + \delta  + 1} \right)_{k - 1} }}{{\left( {\alpha  + 1} \right)_k \left( {\beta  + \delta } \right)_k \left( {\gamma  + 1} \right)_k k!}}}   \label{racah-lhs1}\\ 
&&  \times \left[ {\left( {x + \gamma  + 1} \right)\left( {x + 1} \right)\left( {x + \gamma  + \delta  + k} \right) + \left( {x - \beta  + \gamma  + 1} \right)\left( {k - x - 1} \right)\left( {x + \gamma  + \delta } \right)} \right] \nonumber 
\end{eqnarray}
and
\begin{eqnarray}
&& R_n \left( {\lambda \left( x \right);\alpha ,\beta ,\gamma ,\delta  - 1} \right) - R_n \left( {\lambda \left( {x + 1} \right);\alpha ,\beta ,\gamma ,\delta  - 1} \right) \label{racah-lhs2} \\ 
&&  = \sum\limits_{k = 0}^n {\frac{{\left( { - n} \right)_k \left( {n + \alpha  + \beta  + 1} \right)_k \left( { - x} \right)_{k - 1} \left( {x + \gamma  + \delta + 1 } \right)_{k - 1} }}{{\left( {\alpha  + 1} \right)_k \left( {\beta  + \delta } \right)_k \left( {\gamma  + 1} \right)_k k!}}k\left( {2x + \gamma  + \delta  + 1} \right)} .\nonumber  
\end{eqnarray}
Therefore, combining (\ref{racah-lhs1}) and (\ref{racah-lhs2}), we have the following expression for the left hand side of~(\ref{racah-eq1}):
\begin{eqnarray}
&&\sum\limits_{k = 0}^n {\frac{{\left( { - n} \right)_k \left( {n + \alpha  + \beta  + 1} \right)_k \left( { - x} \right)_{k - 1} \left( {x + \gamma  + \delta  + 1} \right)_{k - 1} }}{{\left( {\alpha  + 1} \right)_k \left( {\beta  + \delta } \right)_k \left( {\gamma  + 1} \right)_k k!}}} \nonumber \\
&&\times \left[ {\beta \left( {k - x - 1} \right)\left( {x + \gamma  + \delta } \right) - k\left( {x + \gamma  + 1} \right)\left( {x + \beta  + \delta } \right)} \right].
\label{racah-lhs3}
\end{eqnarray}
Then, one can rewrite the right hand side of~(\ref{racah-eq1}) as follows:
\begin{eqnarray}
&& \frac{{\left( {n + \alpha  + 1} \right)\left( {n + \beta } \right)}}{{\alpha  + 1}}R_n \left( {\lambda \left( x \right);\alpha  + 1,\beta  - 1,\gamma ,\delta } \right)  \nonumber \\ 
&&  = \frac{{\left( {n + \alpha  + 1} \right)\left( {n + \beta } \right)}}{{\alpha  + 1}}\sum\limits_{k = 0}^n {\frac{{\left( { - n} \right)_k \left( {n + \alpha  + \beta  + 1} \right)_k \left( { - x} \right)_k \left( {x + \gamma  + \delta  + 1} \right)_k }}{{\left( {\alpha  + 2} \right)_k \left( {\beta  + \delta } \right)_k \left( {\gamma  + 1} \right)_k k!}}}   \nonumber \\ 
&&  = \sum\limits_{k = 0}^n {\frac{{\left( { - n} \right)_k \left( {n + \alpha  + \beta  + 1} \right)_k \left( { - x} \right)_k \left( {x + \gamma  + \delta  + 1} \right)_k }}{{\left( {\alpha  + 1} \right)_k \left( {\beta  + \delta } \right)_k \left( {\gamma  + 1} \right)_k k!}}\frac{{\left( {n + \alpha  + 1} \right)\left( {n + \beta } \right)}}{{\alpha  + k + 1}}} \\ 
&&  = \sum\limits_{k = 0}^n \frac{{\left( { - n} \right)_k \left( {n + \alpha  + \beta  + 1} \right)_k \left( { - x} \right)_k \left( {x + \gamma  + \delta  + 1} \right)_k }}{{\left( {\alpha  + 1} \right)_k \left( {\beta  + \delta } \right)_k \left( {\gamma  + 1} \right)_k k!}} \nonumber\\
&& \qquad\qquad \times \left[ {\frac{{\left( {n + \alpha  + \beta  + k + 1} \right)\left( {n - k} \right)}}{{\alpha  + k + 1}} + \left( {\beta  + k} \right)} \right]  \nonumber \\ 
&&  = \sum\limits_{k = 0}^n {\frac{{\left( { - n} \right)_k \left( {n + \alpha  + \beta  + 1} \right)_k \left( { - x} \right)_{k - 1} \left( {x + \gamma  + \delta  + 1} \right)_{k - 1} }}{{\left( {\alpha  + 1} \right)_k \left( {\beta  + \delta } \right)_k \left( {\gamma  + 1} \right)_k k!}}} \nonumber  \\ 
&& \qquad \times \left[ {\left( {\beta  + k} \right)\left( {k - x - 1} \right)\left( {x + \gamma  + \delta  + k} \right) - k\left( {\gamma  + k} \right)\left( {\beta  + \delta  + k - 1} \right)} \right]. \nonumber 
\end{eqnarray}
Now, to prove~(\ref{racah-eq1}), we just need to check that the following equality is correct:
\begin{eqnarray}
\label{racah-lhs4}
&& \beta \left( {k - x - 1} \right)\left( {x + \gamma  + \delta } \right) - k\left( {x + \gamma  + 1} \right)\left( {x + \beta  + \delta } \right) \nonumber \\ 
&&  = \left( {\beta  + k} \right)\left( {k - x - 1} \right)\left( {x + \gamma  + \delta  + k} \right) - k\left( {\gamma  + k} \right)\left( {\beta  + \delta  + k - 1} \right),
\end{eqnarray}
which is obvious. 

The proof of Eq.~(\ref{racah-eq2}) is even simpler than that of Eq.~(\ref{racah-eq1}). 
It is possible to rewrite~(\ref{racah-eq2}) it as follows:
\begin{eqnarray}
\label{racah-eq21}
&& \sum\limits_{k = 0}^n {\frac{{\left( { - n} \right)_k \left( {n + \alpha  + \beta  + 1} \right)_k \left( { - x} \right)_k \left( {x + \gamma  + \delta  + 2} \right)_k }}{{\left( {\alpha  + 1} \right)_k \left( {\beta  + \delta } \right)_k \left( {\gamma  + 1} \right)_k k!}}} \left[ {\left( {x - \alpha  + \gamma  + \delta } \right)\left( {k - x - 1} \right)} \right. \nonumber \\ 
&&  + \left. {\left( {x + \alpha  + 2} \right)\left( {x + \gamma  + \delta  + k + 1} \right) - \left( {2x + \gamma  + \delta  + 2} \right)\left( {\alpha  + k + 1} \right)} \right] = 0. 
\end{eqnarray}
Then (\ref{racah-eq2}) follows from the simple observation that
\begin{eqnarray}
&&\left( {x - \alpha  + \gamma  + \delta } \right)\left( {k - x - 1} \right) + \left( {x + \alpha  + 2} \right)\left( {x + \gamma  + \delta  + k + 1} \right) \nonumber\\
&&= \left( {2x + \gamma  + \delta  + 2} \right)\left( {\alpha  + k + 1} \right).
\end{eqnarray}
\qed
\end{proof}

There are three known cases, when the Racah polynomials $R_n \left( {\lambda \left( x \right);\alpha ,\beta ,\gamma ,\delta } \right)$ reduce to Hahn polynomials $Q_n \left( x;\alpha ,\beta ,m \right)$~\cite[(9.2.15)-(9.2.17)]{koekoek}, defined as
\begin{equation}
Q_n \left( {x;\alpha ,\beta ,m} \right) = {\kern 1pt} _3 F_2 \left( {\begin{array}{*{20}c}
   {\begin{array}{*{20}c}
   { - n,n + \alpha  + \beta  + 1, - x}  \\
   {\alpha  + 1, - m}  \\
\end{array};} & 1  \\
\end{array}} \right).
\end{equation}
For the first two cases, $\left( {\gamma  + 1 =  - m;\delta  \to \infty } \right)$ and $\left( {\delta  =  - \beta  - m - 1;\gamma  \to \infty } \right)$, one recovers a pair of known difference equations for the Hahn polynomials\\
$Q_n \left( x;\alpha+1 ,\beta-1 ,m \right)$ and $Q_n \left( x;\alpha ,\beta ,m \right)$~\cite[(10)-(11)]{stoilova}.
For the third case, ($\alpha  + 1 =  - m;$ $\beta  \to \beta  + \gamma  + m + 1;\delta  \to \infty $) leads to a pair of new difference equations for Hahn polynomials, with a shift in $m$:
\begin{eqnarray}
\label{hahn-eq}
&& \left( {x + 1} \right)Q_n \left( {x;\alpha ,\beta ,m - 1} \right) - \left( {x - m + 1} \right)Q_n \left( {x + 1;\alpha ,\beta ,m - 1} \right)  \nonumber \\
&&\qquad =  m \cdot Q_n \left( {x + 1;\alpha ,\beta ,m} \right), \\
&& m\left( {x - \beta  - m} \right)Q_n \left( {x;\alpha ,\beta ,m} \right) - m\left( {x + \alpha  + 1} \right)Q_n \left( {x + 1;\alpha ,\beta ,m} \right) \nonumber \\
&& \qquad = \left( {n - m} \right)\left( {n + \alpha  + \beta  + m + 1} \right)Q_n \left( {x;\alpha ,\beta ,m - 1} \right).
\end{eqnarray}
Under the limit ($\alpha=pt$;$\beta=(1-p)t$;$t \to \infty$), these equations further reduce to a pair of difference equations for the Krawtchouk polynomials $K_n \left( {x;p,m} \right)$:
\begin{eqnarray}
\label{krawt-eq}
&& \left( {x + 1} \right)K_n \left( {x;p,m - 1} \right) + \left( {m - x + 1} \right)K_n \left( {x + 1;p,m - 1} \right) = m \cdot K_n \left( {x + 1;p,m} \right),\nonumber \\ 
&& m\left( {1 - p} \right)K_n \left( {x;p,m} \right) + m \cdot p \cdot K_n \left( {x + 1;p,m} \right) = \left( {m - n} \right)K_n \left( {x;p,m - 1} \right). 
\end{eqnarray}

Eqs.~(\ref{racah-eq1}) and (\ref{racah-eq2}) can be useful for the construction of finite-discrete quantum oscillator models as well as exactly-solvable spin chains with nearest-neighbour interaction of $m+1$ fermions subject to a zero external magnetic field:
\begin{equation}
\hat H = \sum\limits_{k = 0}^{m - 1} {J_k \left( {a_k^ +  a_{k + 1}  + a_{k + 1}^ +  a_k } \right)} ,
\label{spin chain}
\end{equation}
where, $J_k$ expresses the coupling strength between two neighbour fermions $k$ and $k+1$ and has the following expression:
\begin{equation}
J_k  = \left\{ \begin{array}{l}
 \sqrt {\left( {k + 1} \right)\left( {m - k} \right)f\left( {\alpha ,\beta ,\delta } \right)} ;\qquad\qquad \textrm{k - odd} \\ 
 \sqrt {\left( {k + 2\alpha  + 2} \right)\left( {m - k + 2\beta } \right)g\left( \delta  \right)} ;\quad \textrm{k - even} \\ 
 \end{array} \right. 
\label{coupling}
\end{equation}
with $f\left( {\alpha ,\beta ,\delta } \right)$ and $g\left( \delta  \right)$ defined as follows:
\begin{equation}
f\left( {\alpha ,\beta ,\delta } \right) = \frac{{\left( {k - 2\alpha  + 2\delta  - m} \right)\left( {k + 2\beta  + 2\delta  - 1} \right)}}{{\left( {2k + 2\delta  - m - 1} \right)\left( {2k + 2\delta  - m + 1} \right)}},
\label{f}
\end{equation}
\begin{equation}
g\left( \delta  \right) = \frac{{\left( {k - m + 2\delta  - 1} \right)\left( {k + 2\delta } \right)}}{{\left( {2k + 2\delta  - m - 1} \right)\left( {2k + 2\delta  - m + 1} \right)}}.
\label{g}
\end{equation}

\section{Wilson polynomials as analytical solutions of new difference equations}

The Wilson polynomial $W_n \left( {x^2 ;a,b,c,d} \right)$ of degree $n$ ($n = 0,1, \ldots$) in the variable $x$ is defined by:
\begin{equation}
\label{wilson-def}
\frac{{W_n \left( {x^2 ;a,b,c,d} \right)}}{{\left( {a + b} \right)_n \left( {a + c} \right)_n \left( {a + d} \right)_n }} = {\kern 1pt} _4 F_3 \left( {\begin{array}{*{20}c}
   {\begin{array}{*{20}c}
   { - n,n + a + b + c + d - 1,a + ix,a - ix}  \\
   {a + b,a + c,a + d}  \\
\end{array};} & 1  \\
\end{array}} \right).
\end{equation}
The polynomial satisfies an orthogonality relation in the continuous space $\left[ {0,\left. { + \infty } \right)} \right.$ under the condition ${\mathop{\rm Re}\nolimits} \left( {a,b,c,d} \right) > 0$~\cite[(9.1.2)]{koekoek}.

By putting $\alpha  = a + b - 1$, $\beta  = c + d - 1$, $\gamma  = a + d - 1$, $\delta  = a - d$ and $x \to  - a + ix$ in Eqs. (\ref{racah-eq1}) and (\ref{racah-eq2}) as well as taking into account the duality of Racah polynomials (\ref{racah-def}) in $n$ and $x$, one can transfer them to Wilson polynomials and obtain the following three-term recurrence relations
\begin{eqnarray}
&& W_n \left( {x^2 ;a,b,c,d} \right) = \frac{{n + a + b + c + d - 1}}{{2n + a + b + c + d - 1}}W_n \left( {x^2 ;a,b,c,d + 1} \right) \nonumber \\
&& - \frac{{n(n + a + b - 1)(n + a + c - 1)(n + b + c - 1)}}{{2n + a + b + c + d - 1}}W_{n - 1} \left( {x^2 ;a,b,c,d + 1} \right), \\
&& \left( {x^2  + d^2 } \right)W_n \left( {x^2 ;a,b,c,d + 1} \right) = \nonumber \\
&& \frac{{(n + a + d)(n + b + d)(n + c + d)}}{{2n + a + b + c + d}}W_n \left( {x^2 ;a,b,c,d} \right) \nonumber \\ 
&& \qquad- \frac{1}{{2n + a + b + c + d}}W_{n + 1} \left( {x^2 ;a,b,c,d} \right), 
\end{eqnarray}
and the difference equations:
\begin{eqnarray}
&& \left[ {\frac{{\left( {a + ix} \right)\left( {b + ix} \right)}}{{2ix}}e^{ - {\textstyle{i \over 2}}\partial _x }  - \frac{{\left( {a - ix} \right)\left( {b - ix} \right)}}{{2ix}}e^{{\textstyle{i \over 2}}\partial _x } } \right]W_n \left( {x^2 ;a + {\textstyle{1 \over 2}},b + {\textstyle{1 \over 2}},c,d} \right) \nonumber \\ 
&&  \qquad = \left( {n + a + b} \right)W_n \left( {x^2 ;a,b,c + {\textstyle{1 \over 2}},d + {\textstyle{1 \over 2}}} \right), \label{dif-wil1}\\ 
&& \left[ {\frac{{\left( {c + ix} \right)\left( {d + ix} \right)}}{{2ix}}e^{ - {\textstyle{i \over 2}}\partial _x }  - \frac{{\left( {c - ix} \right)\left( {d - ix} \right)}}{{2ix}}e^{{\textstyle{i \over 2}}\partial _x } } \right]W_n \left( {x^2 ;a,b,c + {\textstyle{1 \over 2}},d + {\textstyle{1 \over 2}}} \right) \nonumber \\ 
&&  \qquad = \left( {n + c + d} \right)W_n \left( {x^2 ;a + {\textstyle{1 \over 2}},b + {\textstyle{1 \over 2}},c,d} \right). \label{dif-wil2}
\end{eqnarray}
Introducing orthonormalized Wilson polynomials, one can reformulate (\ref{dif-wil1}) and (\ref{dif-wil2}) in a more compact form:
\begin{eqnarray}
&& \left[ {e^{ - {\textstyle{i \over 2}}\partial _x }  - \frac{{\left( {a - ix} \right)\left( {b - ix} \right)}}{{2ix}}e^{{\textstyle{i \over 2}}\partial _x } \frac{{\left( {c + ix} \right)\left( {d + ix} \right)}}{{2ix}}} \right]\tilde W_n \left( {x^2 ;{\textstyle{1 \over 2}},0} \right) \nonumber\\
&&\qquad = \sqrt {\left( {n + a + b} \right)\left( {n + c + d} \right)} \tilde W_n \left( {x^2 ;0,{\textstyle{1 \over 2}}} \right), \nonumber \\ 
&& \left[ {e^{ - {\textstyle{i \over 2}}\partial _x }  - \frac{{\left( {c - ix} \right)\left( {d - ix} \right)}}{{2ix}}e^{{\textstyle{i \over 2}}\partial _x } \frac{{\left( {a + ix} \right)\left( {b + ix} \right)}}{{2ix}}} \right]\tilde W_n \left( {x^2 ;0,{\textstyle{1 \over 2}}} \right) \nonumber\\
&&\qquad = \sqrt {\left( {n + a + b} \right)\left( {n + c + d} \right)} \tilde W_n \left( {x^2 ;{\textstyle{1 \over 2}},0} \right), \label{orth-wil}
\end{eqnarray}
where $\tilde W_n \left( {x^2 ;{\textstyle{1 \over 2}},0} \right) \equiv \tilde W_n \left( {x^2 ;a + {\textstyle{1 \over 2}},b + {\textstyle{1 \over 2}},c,d} \right)$ and \\
$\tilde W_n \left( {x^2 ;0,{\textstyle{1 \over 2}}} \right) \equiv \tilde W_n \left( {x^2 ;a,b,c + {\textstyle{1 \over 2}},d + {\textstyle{1 \over 2}}} \right)$. 
As a special case, when $a=c$ and $b=d$, both (\ref{dif-wil1}) and (\ref{dif-wil2}) reduce to difference equations for the continuous dual Hahn polynomials $S_n \left( {4x^2 ;2a,2b,{\textstyle{1 \over 2}}} \right)$~\cite[(9.3.6)]{koekoek}. 
Then, they can be considered as a fermionic extension of the quantum harmonic oscillator model, whose algebra is Lie algebra $su(1,1)$ deformed by a reflection operator~\cite{jafarov1}. 
Under another limit, reducing Wilson polynomials to continuous Hahn polynomials~\cite[(9.1.17)]{koekoek}, one obtains from Eqs.~(\ref{dif-wil1}) and (\ref{dif-wil2}) a pair of difference equations for continuous Hahn polynomials
\begin{eqnarray}
\label{diff-ch1}
 \left[ {\left( {ix + b} \right)e^{ - {\textstyle{i \over 2}}\partial _x }  - \left( {ix - d} \right)e^{{\textstyle{i \over 2}}\partial _x } } \right]p_n \left( {x;0,{\textstyle{1 \over 2}}} \right) = \left( {n + b + d} \right)p_n \left( {x;{\textstyle{1 \over 2}},0} \right), \\ 
\label{diff-ch2}
 \left[ {\left( {ix + a} \right)e^{ - {\textstyle{i \over 2}}\partial _x }  - \left( {ix - c} \right)e^{{\textstyle{i \over 2}}\partial _x } } \right]p_n \left( {x;{\textstyle{1 \over 2}},0} \right) = \left( {n + a + c} \right)p_n \left( {x;0,{\textstyle{1 \over 2}}} \right),
\end{eqnarray}
where, $p_n \left( {x;0,{\textstyle{1 \over 2}}} \right) \equiv p_n \left( {x;a,b + {\textstyle{1 \over 2}},c,d + {\textstyle{1 \over 2}}} \right)$ and \\
$p_n \left( {x;{\textstyle{1 \over 2}},0} \right) \equiv p_n \left( {x;a + {\textstyle{1 \over 2}},b,c + {\textstyle{1 \over 2}},d} \right)$.

Surprisingly, both Eqs.~(\ref{diff-ch1}) and (\ref{diff-ch2}) generalize a difference equation, whose solution is the Meixner-Pollaczek polynomial~\cite[(9.7.5)]{koekoek}. 
Therefore, they can be considered as a fermionic extension of the $su(1,1)$ Meixner-Pollaczek oscillator~\cite{klimyk}.

\section{Conclusion}

Racah and Wilson polynomials, which occupy the top level in the Askey scheme of hypergeometric orthogonal polynomials, are defined through the $ _4 F_3$ type hypergeometric series. 
Under certain conditions, there is a well-known orthogonality relation for the Racah polynomials with respect to a discrete measure as well as for the Wilson polynomials with respect to a continuous measure. 
These polynomials are explicit analytical solutions of known difference equations with quadratic-like eigenvalues. 
In current work, we introduce a pair of novel difference equations or three-term recurrence relations, whose solutions are also expressed in terms of the Racah or Wilson polynomials depending on nature of the finite-difference step. 
The proof of these equations is presented for case of Racah polynomials. 
These equations may turn out to be good candidates for building some new fermionic oscillator models as well as exactly-solvable spin chains with a nearest-neighbour interaction. 
A number of special cases and limit relations are also examined, which allow to introduce similar difference equations for the orthogonal polynomials of the $ _3 F_2$ and $ _2 F_1$ types.

\begin{acknowledgement}
EIJ acknowledges support from Research Grant \textbf{EIF-2012-2(6)-39/08/1} of the Science Development Foundation under the President of the Republic of Azerbaijan.
\end{acknowledgement}

\biblstarthook{}


\begin{thebibliography}{99.}%
%
%
%
\bibitem{koekoek} R. Koekoek, P.A. Lesky, R.F. Swarttouw, \textit{Hypergeometric orthogonal polynomials and their $q$-analogues}, Springer Monographs in Mathematics (Springer-Verlag, Berlin, 2010)
%
\bibitem{landau} L. Landau, E.M. Lifshitz, \textit{Quantum Mechanics: Non-Relativistic Theory}, (Oxford: But\-ter\-worth-Heinemann, 1997)
%
\bibitem{wigner} E.P. Wigner, Phys. Rev.
\textbf{77} (1950) 711.
%
\bibitem{ohnuki} Y. Ohnuki, S. Kamefuchi, \textit{Quantum Field Theory and Parastatistics}, (Springer, New-York, 1982)
%
\bibitem{groenevelt} W. Groenevelt, E. Koelink, J. Phys. A: Math. Gen.
\textbf{35} (2002) 65.
%
\bibitem{jafarov1} E.I. Jafarov, N.I. Stoilova, J. Van der Jeugt, SIGMA
\textbf{8} (2012) 025.
%
\bibitem{jafarov2} E.I. Jafarov, N.I. Stoilova, J. Van der Jeugt, J. Phys. A: Math. Theor.
\textbf{44} (2011) 265203.
%
\bibitem{stoilova} N.I. Stoilova, J. Van der Jeugt, SIGMA
\textbf{7} (2011) 033.
%
\bibitem{klimyk} A.U. Klimyk, Ukr. J. Phys.
\textbf{51} (2006) 1019.
%
\end{thebibliography}
\end{document}